\documentclass[aps,raggedbottom,nobalancelastpage,amssymb,groupeaddress,nofootinbib]{revtex4}

\usepackage[english]{babel}
\usepackage{graphicx}
\usepackage{amssymb}
\usepackage{amsmath}
\usepackage{amsthm}
\usepackage{amscd}
\usepackage{eucal}
\usepackage{color}
\usepackage{bm}

\usepackage{natbib}
\usepackage[bookmarks=true,colorlinks,linkcolor=red,urlcolor=blue,citecolor=blue]{hyperref}
\usepackage{dcolumn}

\usepackage{stmaryrd} 

\def\be{\begin{equation}}
\def\ee{\end{equation}}
\def\bea{\begin{eqnarray}}
\def\eea{\end{eqnarray}}
\def\bsub{\begin{subequations}}
\def\esub{\end{subequations}}

\usepackage{ulem}

\newtheorem{theorem}{Theorem}
\newtheorem{corollary}{Corollary}
\newtheorem{lemma}{Lemma}
\newtheorem{proposition}{Proposition}

\usepackage{csquotes}
\newenvironment{itquote}
{\begin{quote}\itshape}
{\end{quote}}

\begin{document}

\title{Generalized second law of thermodynamics in the Glosten-Milgrom model}

\newcommand{\cfm}{Capital Fund Management, 23 rue de l'Universit\'e, 75007 Paris, France}

\author{Pierre Carmier\footnote{
The views expressed in this publication are those of the author and do not purport to reflect those of Capital Fund Management.}}
\affiliation{\cfm}
\date{\today}

\begin{abstract}
We derive an upper bound for the expected gain of informed traders in the Glosten-Milgrom model with finite horizon, fully analogous to a generalized $2^{nd}$ law of thermodynamics.
This result extends that obtained by Touzo et al. \cite{Touzo} a couple of years ago. 
The proof relies on Bayesian inference (exploiting the invariance of the problem under consecutive game sequences) and an interesting entropic inequality.
We also provide numerical results both supporting the existence of a characteristic timescale in the model and illustrating the magnitude of gain fluctuations.
Other possible extensions are discussed.
\end{abstract}


\maketitle

\section{Introduction}
\label{sec:intro}

In a recent paper \cite{Touzo}, Touzo et al. made a remarkable connection between information thermodynamics and an agent-based toy model describing how information is incorporated into prices.
The latter provides a simplified setting to describe the interaction between liquidity providers and liquidity takers with different information sets \cite{Kyle,GM,book}.
The goal of the present work is to emphasize how this connection can be strengthened by deriving a result fully analogous to a generalized $2^{nd}$ law of thermodynamics.

Thermodynamics is a fascinating topic which came to life in the middle of the $19^{th}$ century with the advent of the industrial revolution, 
beginning with the invention of the steam engine. 
Despite very profound discoveries (quantum mechanics and general relativity) which reshaped our understanding of fundamental physics, 
thermodynamical laws have managed to endure and inspire countless research directions. Their resilience is a testimony to their universality.
The famous $2^{nd}$ law of thermodynamics emphasizes the irreversibility of macroscopic physical processes (arrow of time). 
This is a statistical law which emerges at the macroscopic level from the interaction of a very large number of microsopic degrees of freedom \cite{Diu}.
A recent body of research known as stochastic thermodynamics has been instrumental in explaining the statistical nature of this law by providing a quantitative understanding 
of what happens to "small" systems, materialized in fluctuation theorems \cite{Seifert}. 
More to the point with respect to our objective in this paper, another line of research coined information thermodynamics \cite{IT} has sought to account for information flows in 
thermodynamic systems, thereby helping to solve the infamous paradox of Maxwell's demon.

Among the many different formulations of the $2^{nd}$ law, we shall be interested in the following one:
\begin{itquote}
No transfer of energy, aka work $W$, can be extracted on average from a cyclic transformation of a thermodynamic system $Y$ at constant temperature $T$...

...unless one can acquire some information $I$ on the system through a measurement $M$, in which case $\mathbb{E}[W] \leq TI(Y; M)$.
\end{itquote}
Our main result shall be an equivalent statement $\mathbb{E}[G_n] \leq TI(Y; X_{1:n})$ bounding the gain that an informed trader can extract on average 
from the market using his private information (notations shall be introduced below). 
As suggested in \cite{Touzo}, this can also be interpreted as a generalized no-arbitrage theorem, allowing for arbitrage only if private information is available.

The paper is organized as follows. The stage is set in Section \ref{sec:result}, along with the introduction of notations, and the central inequality is stated.
Proof of the inequality along with intermediate technical results are given in section \ref{sec:proof}.
Subsection \ref{subsec:Bayes} highlights that each step of the game is independent and, thanks to Bayesian inference, that the problem can be fully understood by studying a single step of the game. 
Subsection \ref{subsec:entropy} is devoted to proving the single step inequality which is shown can be reformulated as a statement on differentials of binomial entropy. 
Numerical illustrations are provided in section \ref{sec:Num}, additionally allowing to test various scaling hypotheses.
Finally, we conclude in section \ref{sec:Concl} and offer some perspective on our results.

\section{Notations and main result}
\label{sec:result}

The version of the Glosten-Milgrom model considered here follows closely that presented in \cite{Touzo}, namely a game featuring 3 different players: 
an informed trader, a noise trader and a market maker. 
The game proceeds via sequential (and repeated) interactions between the market maker, who sets break-even transaction prices, 
and traders chosen at random which post orders based on their private information. 
Before stating the inequality, let us start by defining some notations. 
Let $Y\in\{0, 1\}$ be the value of the asset known to the informed trader. 
The market maker's knowledge of this value is encoded in his prior distribution $p(Y) \sim {\cal B}(\theta)$, a Bernoulli distribution with parameter $\theta=p(Y=1)$.
A the beginning of each step $n$ of the game, the market maker posts bid $b_n$ and ask $a_n$ prices which represent his best guess of the asset's value, 
namely $b_n=\mathbb{E}[Y|x_{1:n-1}, X_n=0]$ and $a_n=\mathbb{E}[Y|x_{1:n-1}, X_n=1]$.
Here, $x_{1:n-1} = (x_1,...x_{n-1})$ are the past sequences of orders. 
The informed trader sends orders according to his perfect knowledge of the asset's value, while the noise trader is assumed to randomly buy or sell the asset.
The uncertainty for the market maker lies in the fact that he never knows which of the informed or noise trader is posting the order, 
the latter being chosen with a probability $0\leq \nu\leq 1$ which can also be interpreted as the frequency of informed traders in the population of traders.

As a zero-sum game, an important quantity is how much gain the informed trader is expected to extract from the noise trader. 
The central result of Touzo et al. was that the expected gain of the informed trader is upper bounded as $\mathbb{E}(G) \leq TH(Y)$, where
\be
\label{eq:temp}
T = \left(\frac{1+\nu}{2}\log(\frac{1+\nu}{1-\nu})\right)^{-1}
\ee
is an effective temperature characterizing the bath of noise traders and $H(Y) = -\theta\log{\theta} - (1-\theta)\log(1-\theta)$ 
is the binomial entropy which encodes the amount of privileged information held by the informed trader (or equivalently the amount of igorance of the market maker) at the start of the game. 
These quantities are plotted for illustrative purposes in Figure \ref{fig:entropy}.
\begin{figure}[]
\begin{center}
\includegraphics[angle=0,width=0.7\linewidth]{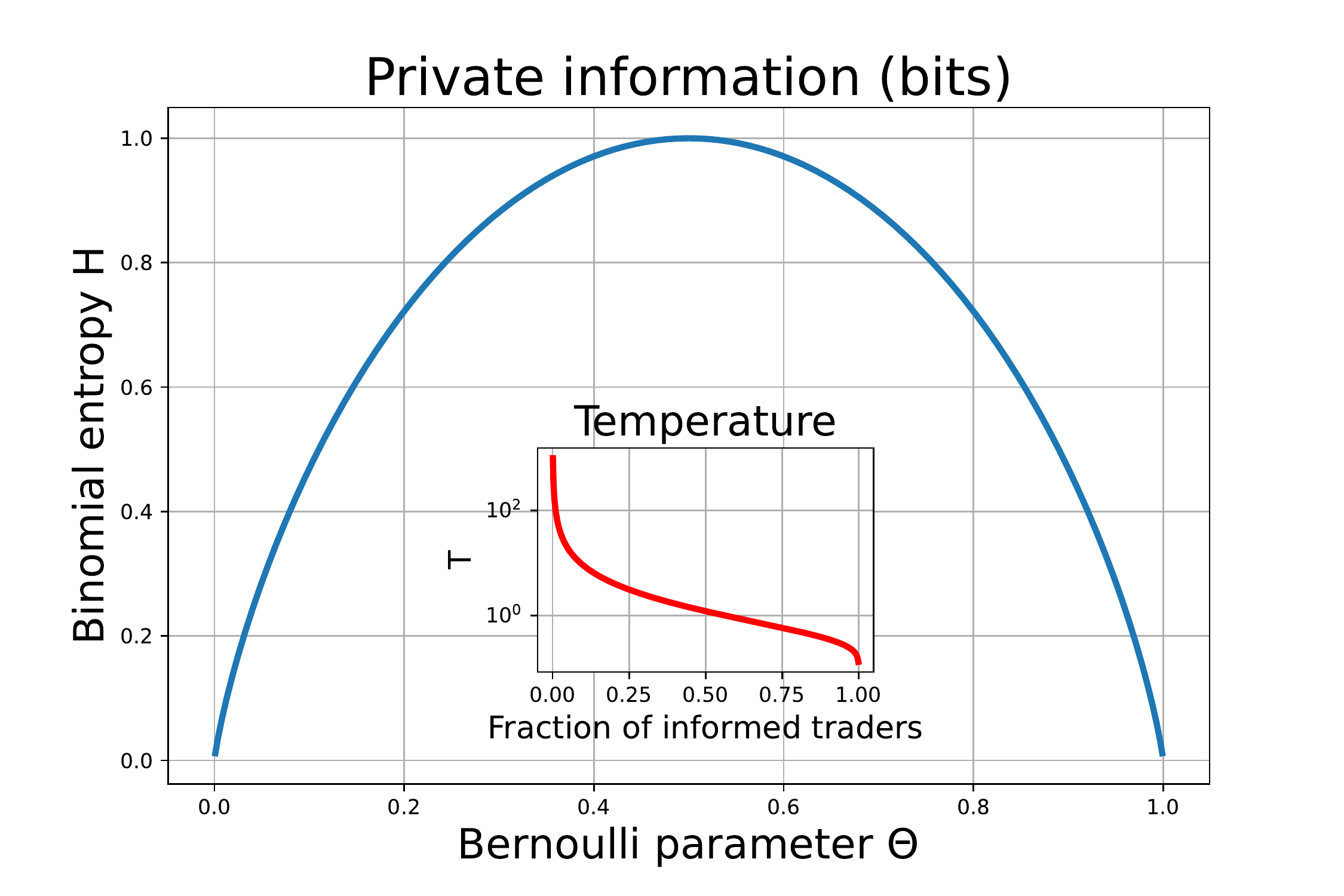}
\caption{The amount of private information possessed by the informed trader is quantified by the entropy of the market maker's prior distribution ${\cal B}(\theta)$, 
which is maximized for $\theta=0.5$. Inset: The effective temperature vanishes in the absence of noise traders and diverges as the fraction of informed traders goes to zero.}
\label{fig:entropy}
\end{center}
\end{figure}
The derivation of Eq.~\ref{eq:temp} can be found in \cite{Touzo} and was another important result.
The gain was defined as the wealth accumulated by the informed trader throughout the whole game, namely $G = \lim_{n\to\infty} G_n$ where 
\be
G_n \; \hat{=} \; \sum_{i=1}^n \mu_i
\ee
is the cumulative gain until time $n$, and the stochastic payoff at a single step $i$ is given by
\be
\mu_i = U_i \left(b_i(1-Y) + (1-a_i)Y\right)
\ee
with $U_i \sim {\cal B}(\nu)$ a random variable characterizing which type of trader is selected at step $i$.
Our main contribution is a tighter bound for the accumulated wealth at any time step $n$ stated below.

\begin{theorem}
\label{th}
Let $G_n=\sum_{i=1}^n \mu_i$ be the gain of the informed trader at time $n$, $I(Y; X_{1:n})$ the mutual information between the asset value $Y$ and the order series $X_{1:n}$, 
and $T$ the temperature given by Eq.~\ref{eq:temp}. Then
\be
\boxed{
\mathbb{E}[G_n] \leq TI(Y; X_{1:n})
}
\ee
with equality reached only as the fraction of informed traders $\nu\to0$.
\end{theorem}
The regime $\nu\to0$ corresponds to that of infinite temperature $T\sim\nu^{-1}$, since the information on the asset value provided by the informed trader is drowned in the ambient noise.
Just as explained in \cite{Touzo}, the saturation of the bound in this limit can be understood as originating from the fact that the orders sent by the informed trader are well separated in time, 
such that the convergence of the market maker's estimation of the asset value occurs adiabatically.
In the opposite regime $\nu\to1$, the temperature goes to zero and equilibrium is reached immediately since the information provided by the informed trader can be incorporated in the price without any interference.

An immediate consequence of Theorem \ref{th} is the result of Touzo et al.:
\begin{corollary}
Let $G$ be the gain of the informed trader over an infinite horizon, $H(Y)$ the entropy of the asset value and $T$ the temperature given by Eq.~\ref{eq:temp}. Then
\be
\mathbb{E}[G] \leq TH(Y)
\ee
with equality only as $\nu\to0$.
\end{corollary}
\begin{proof}
This is obtained asymptotically as $n\to\infty$ from Theorem \ref{th} using that 
\begin{equation}
\begin{aligned}
I(Y; X_{1:n}) \; &\hat{=} \; \int dx_{1:n}dy \; p(y, x_{1:n})\log\left(\frac{p(y, x_{1:n})}{p(y)p(x_{1:n})}\right)
\\
& = -\int dx_{1:n}dy \; p(y, x_{1:n})\log{p(y)} + \int dx_{1:n}p(x_{1:n})\;\int dy \; p(y|x_{1:n})\log{p(y|x_{1:n})}
\\
& = H(Y) - H(Y|X_{1:n})
\end{aligned}
\end{equation}
and $\lim_{n\to\infty}H(Y|X_{1:n})=0$ as the market maker's estimator of the asset value converges.
\end{proof}

\section{Proof}
\label{sec:proof}

We shall prove Theorem \ref{th} in this section using Bayesian inference and a reasonable amount of algebra. 
The first piece of the solution is to realize that Theorem \ref{th} is a direct consequence of the following intermediate result:
\begin{lemma}[L1]
\label{l1}
Let $\mu_i$ be the stochastic payoff of the informed trader at any step $i\geq 1$ and $T$ the temperature.
The following upper bound holds
\be
\mathbb{E}[\mu_i|x_{1:i-1}] \leq TI(Y; X_i|x_{1:i-1})
\ee
where $I(Y; X_i|x_{1:i-1})$ is the mutual information between the asset value $Y$ and the current order $X_i$ conditionally on the realized trajectory of past orders $x_{1:i-1}$.
Equality is reached only as $\nu\to 0$ (with $\mathbb{E}[\mu_i] = O(\nu)$).
\end{lemma}
The result stated in this Lemma is true for any particular order trajectory $x_{1:i-1}$. Taking the expectation over all possible trajectories on both sides yields
\begin{equation}
\mathbb{E}[\mu_i] \leq TI(Y; X_i|X_{1:i-1}) \; .
\end{equation}
From there, the chain rule \cite{Cover} for conditional mutual information $I(Y; X_{1:n}) = I(Y; X_n | X_{1:n-1}) + I(Y; X_{1:n-1})$ allows proving Theorem \ref{th}, as
\be
\begin{aligned}
\mathbb{E}[G_n] & = \sum_{i=1}^n\mathbb{E}[\mu_i] 
\\
& \leq T\sum_{i=1}^nI(Y; X_i | X_{1:i-1}) = T \sum_{i=1}^n \left(I(Y; X_{1:i}) - I(Y; X_{1:i-1})\right) = TI(Y; X_{1:n}) \; .
\end{aligned}
\ee
To prove Lemma \ref{l1}, we shall need to explicitate how the various quantities entering the inequality depend on the parameters of the Glosten-Milgrom model. 

\subsection{Bayesian inference}
\label{subsec:Bayes}

Let us start by observing that the market maker's estimation problem is very naturally framed as a Bayesian inference problem.
From that perspective, the market maker's knowledge of $Y$ at step $n$ can be encoded in his posterior distribution which simply updates the parameter $\theta$ accordingly. 
Let us illustrate how this works for the first step of the game. The likelihood for the market maker to observe an order $x$ given the asset value $y$ is 
\be
{\cal L}_{xy} \; \hat{=} \; p(X_1=x|Y=y)=\frac{1-\nu}{2}+\nu\delta_{x,y} \; ,
\ee
where the first term reflects the noise trader's absence of preference, while the second reflects the informed trader's perfect knowledge.
Using Bayes rule, the market maker's posterior distribution follows as 
\be
\label{eq:post}
p(Y=1|X_1=x) = \frac{p(X_1=x|Y=1)p(Y=1)}{p(X_1=x|Y=1)p(Y=1)+p(X_1=x|Y=0)p(Y=0)} = \frac{\theta {\cal L}_{x1}}{\theta {\cal L}_{x1} + (1-\theta)(1-{\cal L}_{x1})} \; .
\ee
Basically, the order received by the market maker acts as an informative measurement on the value of the asset which allows the market maker to update his belief on the asset's value 
$\theta\to \theta_1 \; \hat{=} \; \mathbb{E}[Y|x_1]=p(Y=1|x_1)$.
This result can be formalized as follows
\begin{proposition}[P1]
Given a Bernoulli prior distribution $p(Y)\sim{\cal B}(\theta)$ on the asset value $Y$ and a likelihood ${\cal L}_{xy}=\frac{1-\nu}{2} + \nu\delta_{xy}$, the posterior distribution obeys $p(Y|x_{1:n})\sim{\cal B}(\theta_n)$, 
where $\theta_n$ depends on $\nu$, $x_n$ and $\theta_{n-1}$.
\end{proposition} 
\begin{proof}
This holds true for $n=1$ (as displayed in Eq.~\ref{eq:post}) and trivially extends to arbitrary $n\geq1$ by recursion:
\be
\begin{aligned}
\theta_n \; \hat{=} \; p(Y=1|x_{1:n}) & = \frac{p(X_n=x_n|Y=1)p(Y=1|x_{1:n-1})}{p(X_n=x_n|Y=1)p(Y=1|x_{1:n-1})+p(X_n=x_n|Y=0)p(Y=0|x_{1:n-1})} 
\\
& = \frac{\theta_{n-1} {\cal L}_{x_n1}}{\theta_{n-1} {\cal L}_{x_n1} + (1-\theta_{n-1})(1-{\cal L}_{x_n1})} \; .
\end{aligned}
\ee
\end{proof}

Using Bayesian parlance, the conjugacy of the market maker's prior and the likelihood (both Bernoulli) ensure that the posterior distribution remains a Bernoulli.
An important observation is that this is the only thing that changes from one step to another, such that each step can be regarded as independent from the previous one, 
using the updated parameter $\theta_n$ associated with the variable of interest $Y|x_{1:n}$ (by convention $\theta_0=\theta$). 
Note that $\theta_n$ is a function of the realized order trajectory $x_{1:n}$ and, as such, is a random variable.
In fact $\theta_n=\mathbb{E}[Y|x_{1:n}]$ is nothing else but the price of the asset at step $n$.
It is actually quite simple to analytically obtain its probability distribution, as we shall show in Section \ref{sec:Num}, but we do not need it to prove Lemma \ref{l1}.

\begin{proposition}[P2]
Let $p(Y|x_{1:n-1})\sim{\cal B}(\theta_{n-1})$ and denote $q\;\hat{=}\;\frac{1+\nu}{2}$ and $z_{n-1}\;\hat{=}\;q\theta_{n-1}+(1-q)(1-\theta_{n-1})$.
Finally, let 
\begin{align*}
h \colon [0, 1] &\to [0, \log{2}]
\\
x & \mapsto -x\log{x} - (1-x)\log(1-x)
\end{align*}
be the binomial entropy function. Lemma \textbf{(L1)} is true iff so is the following inequality:
\be
\forall q \in [\frac{1}{2}, 1] \;\; \forall \theta_{n-1} \in [0, 1] \;\; \frac{\theta_{n-1}(1-\theta_{n-1})(2q-1)(1-q)}{z_{n-1}(1-z_{n-1})} \leq \frac{h(z_{n-1})-h(q)}{q\log\left(\frac{q}{1-q}\right)}
\ee
\end{proposition}
\begin{proof}
Using $p(Y|x_{1:n-1})\sim{\cal B}(\theta_{n-1})$, let us make the inequality from Lemma \ref{l1} more explicit by evaluating its various terms.
Starting with the expected payoff at step $n$, 
\be
\begin{aligned}
\mathbb{E}[\mu_n|x_{1:n-1}] & = \nu\left( b_n(1-\mathbb{E}[Y|x_{1:n-1}])+(1-a_n)\mathbb{E}[Y|x_{1:n-1}] \right)
\\
& = \nu \left(b_n(1-\theta_{n-1})+(1-a_n)\theta_{n-1} \right) \; ,
\end{aligned}
\ee
this requires expressing bid and ask prices:
\be
\label{eq:bid}
b_n = p[Y=1|x_{1:n-1}, X_n=0] = \frac{(1-\nu)\theta_{n-1}}{(1-\nu)\theta_{n-1}+(1+\nu)(1-\theta_{n-1})} \; ,
\ee
and
\be
\label{eq:ask}
a_n = p[Y=1|x_{1:n-1}, X_n=1] = \frac{(1+\nu)\theta_{n-1}}{(1+\nu)\theta_{n-1}+(1-\nu)(1-\theta_{n-1})} \; .
\ee
Note that the bid-ask spread
\begin{equation}
\begin{aligned}
s_n \; & \hat{=} \; a_n-b_n 
\\
& = \frac{4\nu\theta_{n-1}(1-\theta_{n-1})}{\left((1+\nu)\theta_{n-1}+(1-\nu)(1-\theta_{n-1})\right)\left((1-\nu)\theta_{n-1}+(1+\nu)(1-\theta_{n-1})\right)}
\end{aligned}
\end{equation}
is positive as it should and shrinks as the market maker's estimator $\theta_{n-1}$ converges. Plugging this in the expression for the expected payoff at step $n$, 
one obtains quite remarkably that the latter is simply proportional to the bid-ask spread
\be
\mathbb{E}[\mu_n|x_{1:n-1}] = \frac{1-\nu}{2}s_n \; .
\ee
The other term appearing in the above inequality is the conditional mutual information, which is given by 
\be
\begin{aligned}
I(Y; X_n | x_{1:n-1}) & = \sum_{x, y}p(Y|x_{1:n-1}=y, X_n=x)\log\left(\frac{p(Y|x_{1:n-1}=y, X_n=x)}{p(Y|x_{1:n-1}=y)p(X_n=x)}\right)
\\
& = \sum_{x, y}{\cal L}_{xy}p(Y=y|x_{1:n-1})\log\left(\frac{{\cal L}_{xy}}{p(X_n=x)}\right)
\\
& = \frac{1+\nu}{2}(1-\theta_{n-1})\log\left(\frac{(1+\nu)}{(1+\nu)(1-\theta_{n-1})+(1-\nu)\theta_{n-1}}\right)
\\
& + \frac{1-\nu}{2}\theta_{n-1}\log\left(\frac{(1-\nu)}{(1+\nu)(1-\theta_{n-1})+(1-\nu)\theta_{n-1}}\right)
\\
& + \frac{1-\nu}{2}(1-\theta_{n-1})\log\left(\frac{(1-\nu)}{(1-\nu)(1-\theta_{n-1})+(1+\nu)\theta_{n-1}}\right)
\\
& + \frac{1+\nu}{2}\theta_{n-1}\log\left(\frac{(1+\nu)}{(1-\nu)(1-\theta_{n-1})+(1+\nu)\theta_{n-1}}\right) \; .
\end{aligned}
\ee
Introducing new variables $q=(1+\nu)/2$ and $z_{n-1}=q\theta_{n-1}+(1-q)(1-\theta_{n-1})$, this cumbersome expression simplies considerably as
\be
I(Y; X_n | x_{1:n-1}) = h(z_{n-1}) - h(q)
\ee
where $h \colon x \mapsto -x\log{x}-(1-x)\log(1-x)$ is the binomial entropy function. Likewise, the expected payoff becomes
\be
\mathbb{E}[\mu_n| x_{1:n-1}] = \frac{(2q-1)(1-q)\theta_{n-1}(1-\theta_{n-1})}{z_{n-1}(1-z_{n-1})}
\ee
and the inverse temperature $T^{-1} = q\log(q/(1-q))$. Putting everything together, we arrive at the desired statement from \textbf{(P2)}.
\end{proof}

Along the way, we found the interesting result according to which the informed trader's payoff is proportional to the bid-ask spread set by the market maker. 
Another way to recover this result is to compute the expected gain of the noise trader which is simply $(1-\nu)(\frac{Y-a_n}{2}+\frac{b_n-Y}{2})=-\frac{1-\nu}{2}s_n$.
In other words, trading without any information leads to paying half the spread on average to the market maker. Since the latter breaks even in this model, the result follows.
As a consequence, informed traders can still make money if the market maker earns an additional fraction $m$ of the spread as a fee, provided $\nu m<\frac{1-\nu}{2}$. 
In particular, a single informed trader can choose his trading frequency to be $\nu<(1+2m)^{-1}$ to ensure he makes a profit. 
The total gain achieved will be reduced by a factor $1-\frac{2q-1}{1-q}m$.

\subsection{Entropic inequality}
\label{subsec:entropy}

The inequality only depends on underlying parameters $\nu$ and $\theta_{n-1}$.
Thus, as explained in the previous section, if it can be proven irrespective of the value of $\theta_{n-1}$, it shall hold for any step which is why we now focus on step $n=1$.

\begin{lemma}[L2]
\label{lem2}
Let $q\in[\frac{1}{2}, 1]$, $\theta\in[0,1]$ and $z=q\theta+(1-q)(1-\theta)$. Let $h$ be the binomial entropy function. The following inequality holds:
\be
\frac{\theta(1-\theta)(2q-1)(1-q)}{z(1-z)} \leq \frac{h(z)-h(q)}{q\log\left(\frac{q}{1-q}\right)}
\ee
with equality only as $q\to\frac{1}{2}$.
\end{lemma}

The inequality is depicted in Figure \ref{fig:inequality}.
\begin{figure}[]
\begin{center}
\includegraphics[angle=0,width=0.7\linewidth]{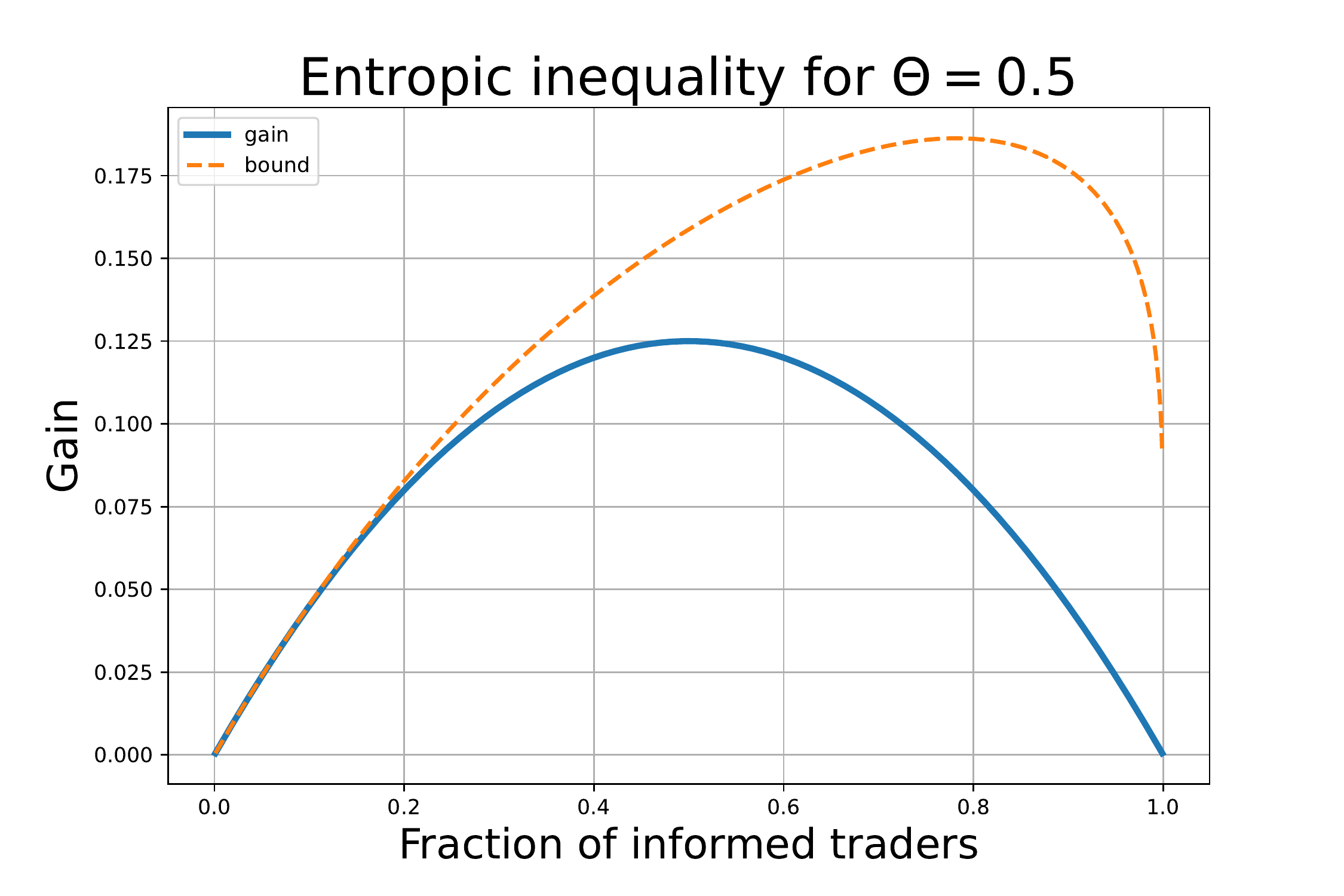}
\caption{Single step expected gain vs thermodynamic bound as a function of $\nu$ for $\theta=0.5$. The tightness of the bound as $\nu\to0$ is apparent.}
\label{fig:inequality}
\end{center}
\end{figure}
Note that both $2q-1$ and $\log(q/(1-q))$ are positive since $1/2\leq q \leq 1$.
By passing the denominator on the right-hand side to the left-hand side and grouping terms, this becomes
\be
(2q-1)\log\left(\frac{q}{1-q}\right)\frac{q(1-q)\theta(1-\theta)}{z(1-z)} \leq h(z) - h(q) \; .
\ee
This expression is invariant under $\theta\to 1-\theta$ and $q\to 1-q$ transformations (noting that $z\to 1-z$ in both cases), 
which is why we choose to restrict $1/2\leq \theta \leq 1$ such that $1/2\leq z \leq q$.

\begin{proposition}[P3]
Let $h$ be the binomial entropy function and let $h^{(i)}=\frac{d^ih}{dx^i}$ be the order $i$ derivative.
Lemma \textbf{(L2)} is true iff the following entropic inequality holds, $\forall q \in [\frac{1}{2}, 1] \;\; \forall z \in [\frac{1}{2}, q]$
\be
\frac{h^{(1)}(q)}{h^{(3)}(q)}\left(h^{(2)}(z)-h^{(2)}(q)\right) \leq h(z) - h(q)
\ee
or equivalently 
\be
f(z) \geq f(q)
\ee
where $f \colon x \mapsto -h^{(3)}(q)h(x) + h^{(1)}(q)h^{(2)}(x)$.
\end{proposition}
\begin{proof}
Starting from the inequality in Lemma \ref{lem2} and using that $\theta(1-\theta)(2q-1)^2 = (z(1-z)-q(1-q))$, the inequality can be reformulated as 
\be
\frac{q^2(1-q)^2}{2q-1}\log\left(\frac{q}{1-q}\right)\left(\frac{1}{q(1-q)}-\frac{1}{z(1-z)}\right) \leq h(z)-h(q) \; .
\ee
Interestingly, the left-hand side can be identified with various derivatives of the binomial entropy
\be
\begin{aligned}
& h^{(1)}(x) = -\log\left(\frac{x}{1-x}\right)
\\
& h^{(2)}(x) = - \frac{1}{x(1-x)}
\\
& h^{(3)}(x) = - \frac{2x-1}{x^2(1-x)^2}
\\
& h^{(4)}(x) = -2\frac{1-3x+3x^2}{x^3(1-x)^3}
\end{aligned}
\ee
where the last one is introduced for further reference. Note that all these derivatives are negative on the interval $[1/2, 1]$. This yields
\be
\frac{h^{(1)}(q)}{h^{(3)}(q)}\left(h^{(2)}(z)-h^{(2)}(q)\right) \leq h(z) - h(q)
\ee
which is some kind of statement on the differential analysis of the binomial entropy function. 
A natural assumption might be that this can be proven somehow using the function's concavity, but it turns out the inequality is tighter than that.
Instead, introducing the auxiliary function $f \colon x \mapsto -h^{(3)}(q)h(x) + h^{(1)}(q)h^{(2)}(x)$, the inequality simply becomes $f(z) \geq f(q)$.
\end{proof}

As a consequence, Lemma \ref{lem2} is proven in particular if the following proposition holds.
\begin{proposition}[P4]
Let $f \colon x \mapsto -h^{(3)}(q)h(x) + h^{(1)}(q)h^{(2)}(x)$, where $h$ is the binomial entropy function. Then
\be
\forall q \in [\frac{1}{2}, 1] \;\; \forall z \in [\frac{1}{2}, q] \;\; f^{(1)}(z) \leq 0 \; .
\ee
\end{proposition}
\begin{proof}
Deriving this function yields
\be
f^{(1)}(z) = h^{(3)}(q)h^{(3)}(z)\left(\frac{h^{(1)}(q)}{h^{(3)}(q)} - \frac{h^{(1)}(z)}{h^{(3)}(z)}\right)
\ee
where the prefactor is positive. So negativity of the derivative is achieved in particular if the function 
\be
z \mapsto \frac{h^{(1)}(z)}{h^{(3)}(z)} = \frac{z^2(1-z)^2}{2z-1}\log\left(\frac{z}{1-z}\right)
\ee 
also has a negative derivative on the interval $[1/2, q]$. Simple algebra shows that
\be
\begin{aligned}
\left(\frac{h^{(1)}(z)}{h^{(3)}(z)}\right)^{(1)} & = \frac{h^{(2)}(z)h^{(3)}(z) - h^{(1)}(z)h^{(4)}(z)}{h^{(3)}(z)^2}
\\
& = \frac{z(1-z)}{(2z-1)^2}g(z)
\end{aligned}
\ee
where the prefactor is positive and $g \colon z \mapsto 2z-1 - 2(1-3z+3z^2)\log(z/(1-z))$. Finally, observe that
\be
g^{(1)}(z) = -6(2z-1)\log\left(\frac{z}{1-z}\right) - 2\frac{1-4z+4z^2}{z(1-z)} \leq 0
\ee
since $1-4z+4z^2$ is positive for $z\geq 1/2$. Given that $g(1/2)=0$, this implies that $g(z)\leq 0$ which allows concluding.
\end{proof}

\section{Numerical results}
\label{sec:Num}

Let us now provide some illustration of our results by simulating the Glosten-Milgrom model numerically.
This will also turn out to be useful to better understand the scaling properties of certain variables of interest.

\subsection{Price trajectories}

The starting point is to gain access to the statistics of price trajectories. There are 2 ways to proceed. 
The first is to generate a large number $N$ of market order trajectories $\{x_{1:n}^{(j)}\}_{1\leq j\leq N}$ of length $n$ for different values of $\nu$ and compute empirical moments.
For example, the expectation of the price at step $i$, which is trajectory-dependent as $\theta_i^{(j)}(\nu) = \theta_i\left(x_{1:i}^{(j)}, \nu\right)$,
can be estimated using the empirical average
\be
\bar{\theta_i}(\nu) = \frac{1}{N}\sum_{j=1}^N\theta_i^{(j)}(\nu) \; .
\ee
\begin{figure}[]
\begin{center}
\includegraphics[angle=0,width=0.7\linewidth]{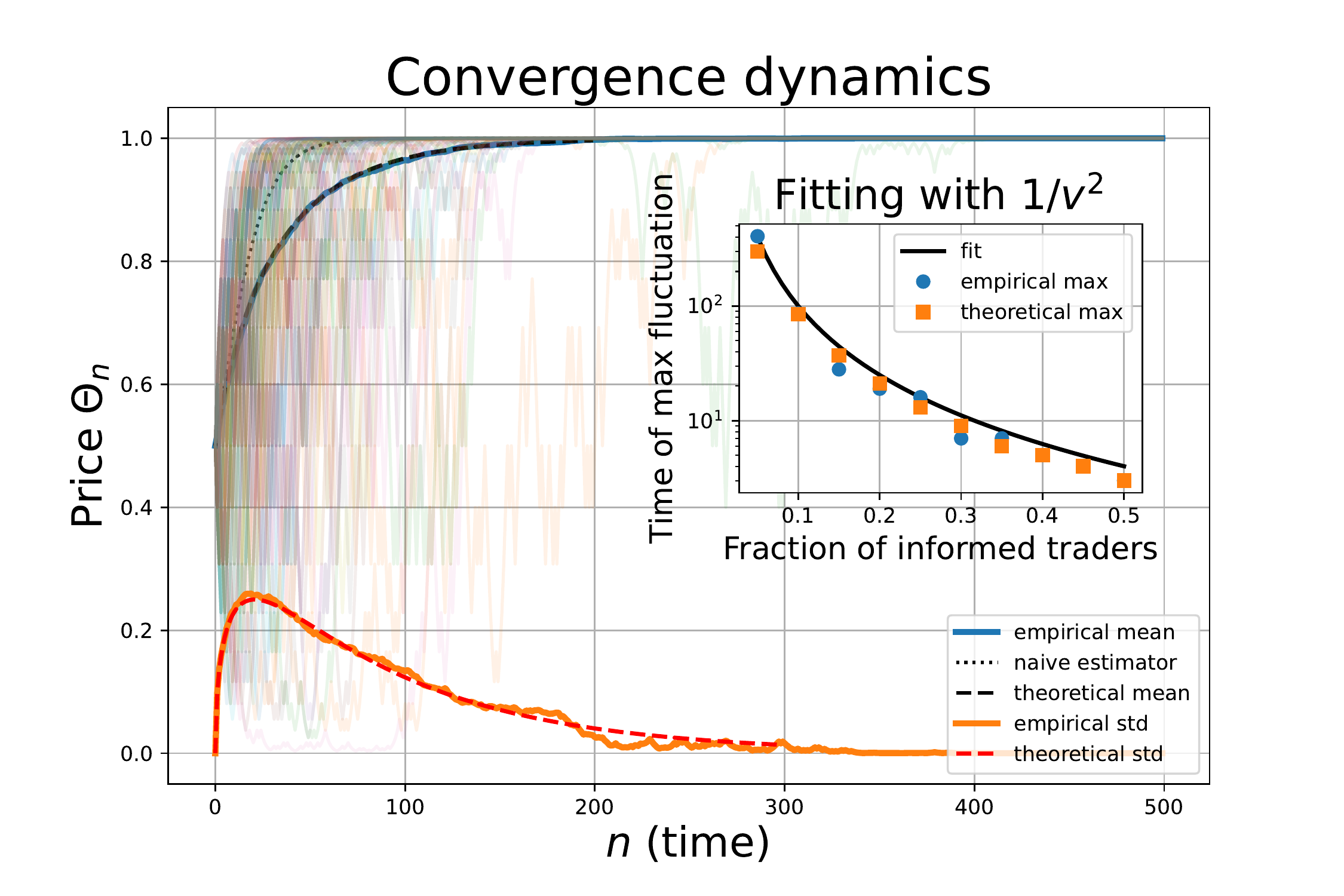}
\caption{Convergence of the market maker's estimator $\theta_n$ of the asset value as a function of time, starting from $\theta=0.5$ and assuming $Y=1$. 
Theoretical expectations are in good agreement with empirical results (obtained using $N=1000$ sample trajectories, 100 of which are displayed in the background for illustration).
Inset: The estimator's variance is maximized at a characteristic time, which seems to scale as $\nu^{-2}$.}
\label{fig:convergence}
\end{center}
\end{figure}
The second is to realize that the probability distribution of $\theta_i$, which encompasses the statistics of price trajectories and on which everything relies, 
can in fact be computed exactly using the map
\be
\begin{aligned}
{\cal L}_q \;\colon\; & [0, 1] \to [0, 1]
\\
& \theta \mapsto \frac{q\theta}{q\theta + (1-q)(1-\theta)} \; ,
\end{aligned}
\ee
recalling that $q=(1+\nu)/2$. 
This map features several nice properties:
\begin{itemize}
\item ${\cal L}_q$ is strictly increasing: $\theta\leq{\cal L}_q[\theta]\leq1$.
\item ${\cal L}_q$ is invertible: ${\cal L}_q^{-1} = {\cal L}_{1-q}$. 
\item When iterated, ${\cal L}_q$ converges exponentially fast to its fixed point:
\be
\begin{aligned}
{\cal L}_q^n[\theta] & \; \hat{=} \; \overbrace{ {\cal L}_q\circ\; ... \; \circ{\cal L}_q}^n \; [\theta]
\\
& = \frac{q^n\theta}{q^n\theta + (1-q)^n(1-\theta)}
\\
& = \frac{\theta}{\theta + (1-\theta)e^{-\frac{n}{\tau_q}}}
\end{aligned}
\ee
with $\tau_q \; \hat{=} \; qT$ the convergence timescale. This can be easily proven recursively.
\end{itemize}
The next proposition puts it to good use.

\begin{proposition}[P5]
Consider the map ${\cal L}_q \colon x \mapsto \frac{qx}{qx + (1-q)(1-x)}$ with $\frac{1}{2}\leq q\leq 1$. 
Then, given an initial condition $\theta\in[0, 1]$, the price at step $n\geq1$ is given by
\be
\theta_n(\beta) = {\cal L}^{2\beta-n}_q\left[\theta\right]
\ee
where $\beta|Y \sim {\cal B}_n\left(1-q+(2q-1)Y\right)$ is a binomial random variable. 
\end{proposition}
\begin{proof}
Denote $\beta=\sum_{i=1}^n\delta_{x_i, 1}$ the number of buy orders. 
The probability for a buy order is a Bernoulli ${\cal B}(1-q+(2q-1)Y)$, thus $\beta\sim{\cal B}_n(1-q+(2q-1)Y)$.
As can be seen from Eq.~\ref{eq:post}, every such order is associated with an iteration of the map ${\cal L}_q$.
In contrast, each of the $n-\beta$ sell orders is associated with an iteration of the map ${\cal L}_{1-q}$.
Using the invertibility of the map, the iterations commute, making the order in which they occur irrelevant. 
The result follows immediately: $\theta_n = {\cal L}_q^\beta\circ{\cal L}_q^{-(n-\beta)}[\theta] = {\cal L}_q^{2\beta-n}[\theta]$.
\end{proof}

\begin{figure}[]
\begin{center}
\includegraphics[angle=0,width=0.7\linewidth]{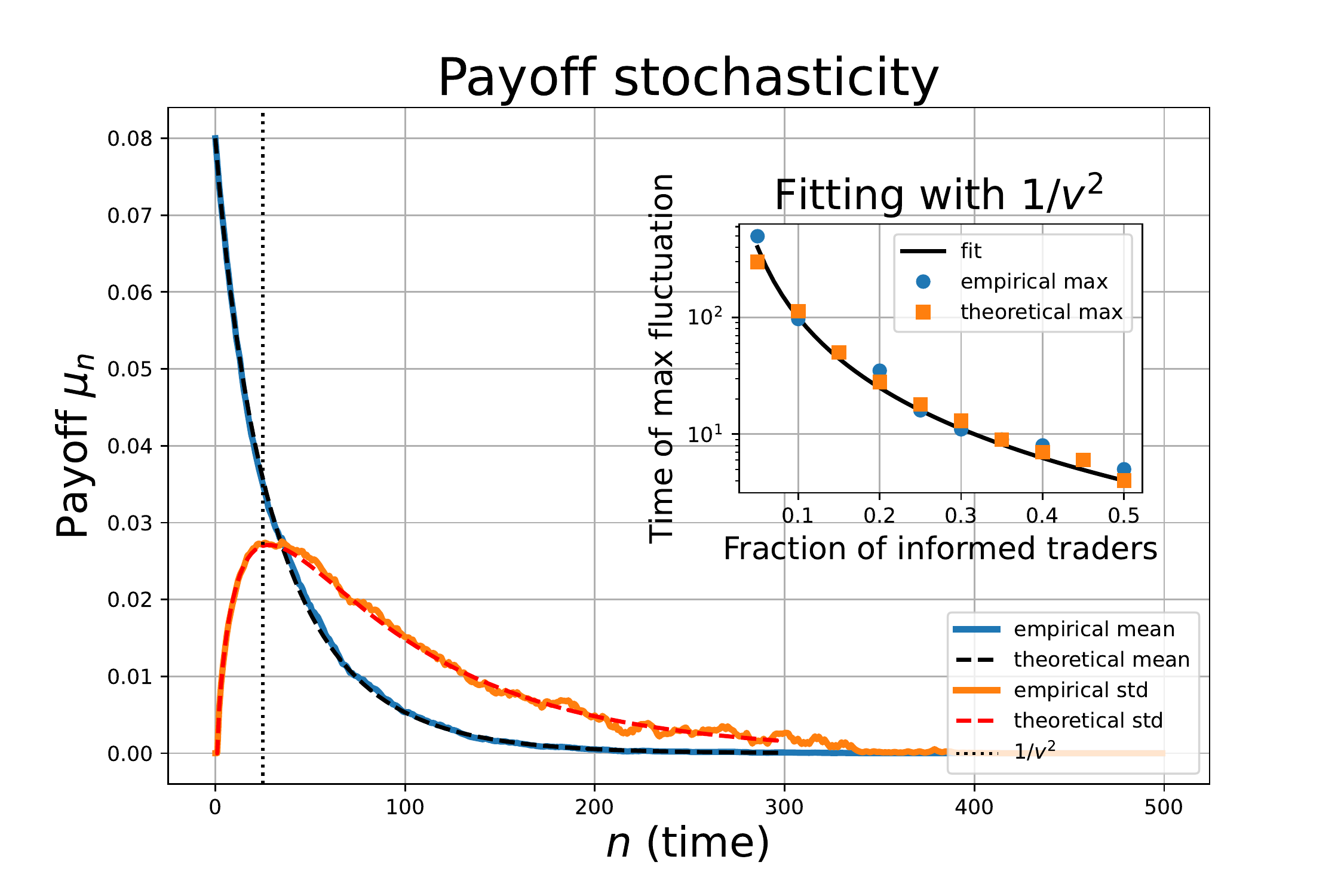}
\caption{Statistics of the informed trader's payoff $\mu_n$ as a function of time, assuming $\theta=0.5$ and $\nu=0.2$. 
The expected payoff decays exponentially with time.
Inset: Payoff variance seems to be maximized when the price set by the market maker fluctuates most.}
\label{fig:payoff}
\end{center}
\end{figure}
The price trajectory as a function of the incoming market orders can thus be seen as randomly flowing towards the map's fixed point with a Brownian drift $q-(1-q)=\nu$.
The expected value of $\theta_n$, obtained by averaging over the binomial distribution of $\beta$
\be
\mathbb{E}[\theta_n|Y=1] = \int d\beta\; p(\beta|Y=1)\theta_n(\beta) = \sum_{k=0}^n C^k_n q^k(1-q)^{n-k}{\cal L}_q^{2k-n}[\theta] \; ,
\ee
is plotted alongside the empirical average in Figure \ref{fig:convergence} as a function of $n$ for $\nu=0.2$. 
Note that the naive estimator ${\cal L}_q^{(2q-1)n}[\theta]$, replacing $\beta$ by its expectation $\mathbb{E}[\beta|Y=1] = nq$, does not work.
On the other hand, the characteristic time at which fluctuations $\mathbb{V}[\theta_n]$ are maximized seems to scale as $\nu^{-2}$, which does match the naive guess
\be
\frac{2\mathbb{E}[\beta]-n}{qT} = \frac{2q-1}{qT}n \sim_{\nu\to0} 2\nu^2n \; .
\ee

\subsection{Payoff statistics}

Next, we move to the stochastic payoff $\mu_i = (2q-1)\left(b_i(1-Y) + (1-a_i)Y\right)$ which depends on bid and ask prices. 
Taking a closer look at Eqs.~(\ref{eq:bid}, \ref{eq:ask}), we find that they can be easily expressed in terms of the map we introduced as
\be
\begin{aligned}
& b_{i+1} = {\cal L}_{1-q}[\theta_i] = {\cal L}_q^{2\beta-i-1}[\theta]
\\
& a_{i+1} = {\cal L}_q[\theta_i] = {\cal L}_q^{2\beta-i+1}[\theta]
\end{aligned}
\ee
which makes sense. As a consequence,
\be
\begin{aligned}
\mathbb{E}[\mu_{i+1}] & = \sum_{y\in\{0, 1\}} p(y)\int d\beta \; p(\beta|Y=y) \; \mu_{i+1}(\beta)
\\
& = (2q-1)\left((1-\theta)\sum_{k=0}^i C_k^i (1-q)^k q^{i-k}{\cal L}_q^{2k-i-1}[\theta] + \theta\sum_{k=0}^i C_k^i q^k (1-q)^{i-k}(1-{\cal L}_q^{2k-i+1}[\theta])\right)
\end{aligned}
\ee
which coincides with the expression obtained by Touzo et al. (reformulated in terms of ${\cal L}_q$). 
\begin{figure}[]
\begin{center}
\includegraphics[angle=0,width=0.7\linewidth]{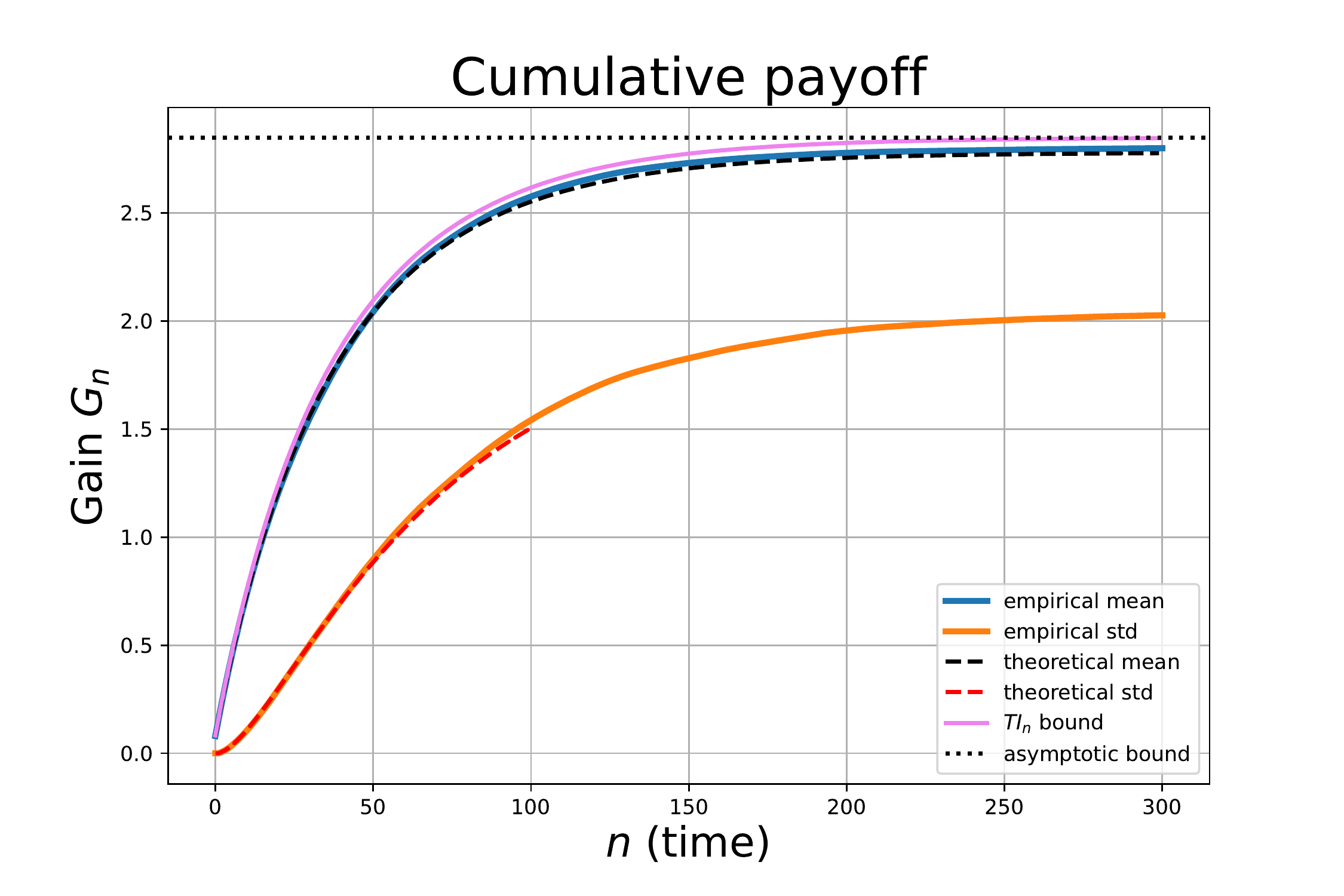}
\caption{Statistics of the informed trader's cumulated payoff $G_n=\sum_{i\leq n}\mu_i$ as a function of time, assuming $\theta=0.5$ and $\nu=0.2$. 
The expected gain is upper bounded at all times according to the expression from Theorem \ref{th}.
Gain fluctuations remain large asymptotically as a consequence of the positive autocorrelation of the payoff.}
\label{fig:gm_gain}
\end{center}
\end{figure}

Alternatively, we have already seen that $\mathbb{E}[\mu_{i+1}|x_{1:i}] = (1-q)s_{i+1}(\theta_i)$. 
It is not difficult to show that the spread can be expressed as a function of the random variable $\beta$ as
\be
s_{i+1}(\beta) = \frac{2q-1}{(1-q)^2}{\cal L}_q^{2\beta-i-1}[\theta]\left(1-{\cal L}_q^{2\beta-i+1}[\theta]\right) \; ,
\ee
yielding a different (yet equivalent) formulation for the expected payoff:
\be
\begin{aligned}
\mathbb{E}[\mu_{i+1}] & = (1-q)\left(\theta\sum_{k=0}^iC_i^kq^k(1-q)^{i-k}s_{i+1}(k) + (1-\theta)\sum_{k=0}^iC_i^k(1-q)^kq^{i-k}s_{i+1}(k)\right) 
\\
& = \frac{2q-1}{1-q}\theta\sum_{k=0}^iC_i^kq^k(1-q)^{i-k}\frac{{\cal L}_q^{2k-i-1}[\theta]\left(1-{\cal L}_q^{2k-i+1}[\theta]\right)}{{\cal L}_q^{2k-i}[\theta]} \; .
\end{aligned}
\ee
Payoff statistics are displayed in Figure \ref{fig:payoff} and seem to inherit some of the scaling properties of the price $\theta_n$.
\begin{figure}[]
\begin{center}
\includegraphics[angle=0,width=0.7\linewidth]{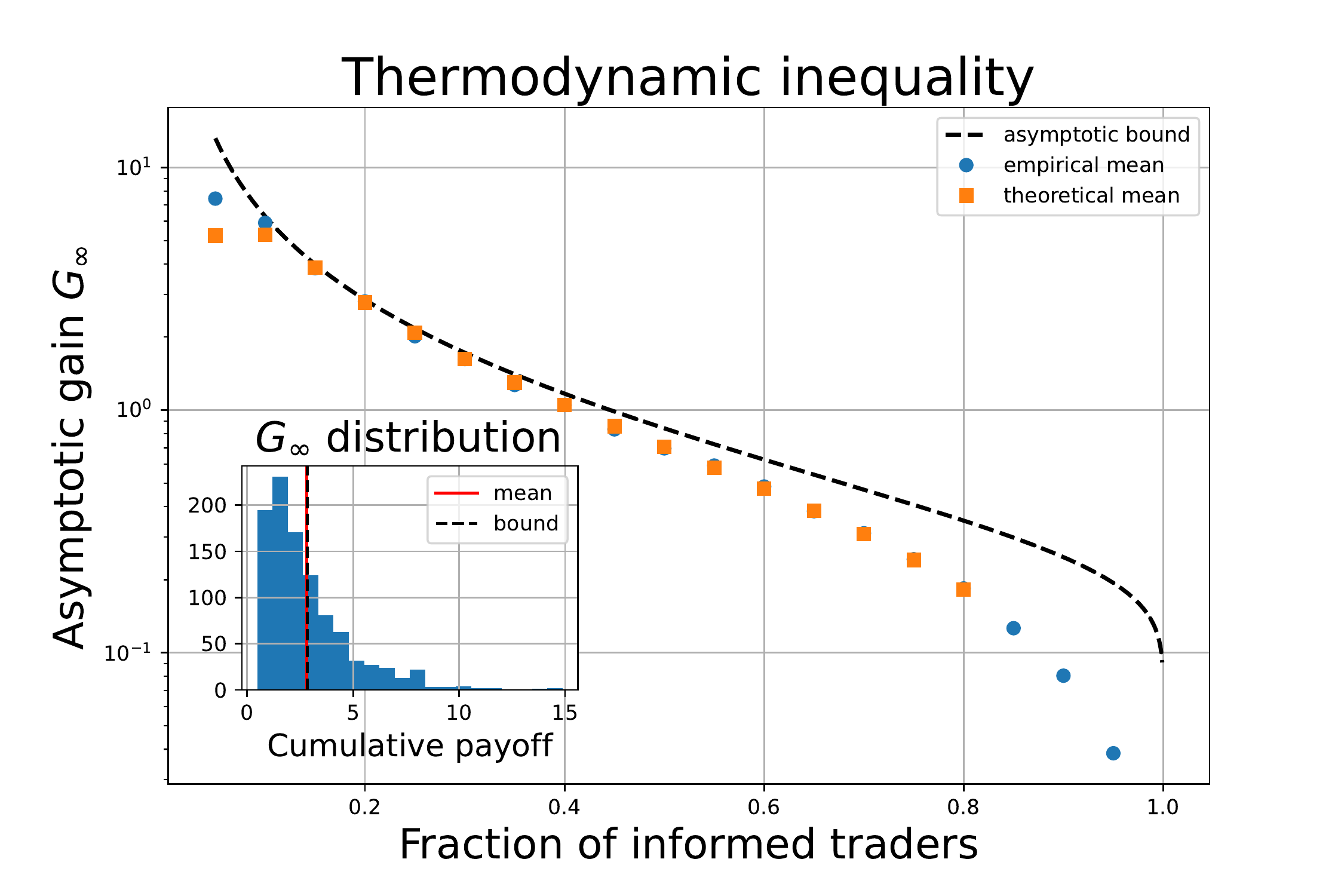}
\caption{Asymptotic expected gain vs $TH(Y)$ as a function of $\nu$, assuming $\theta=0.5$. 
Inset: Empirical histogram of the asymptotic gain for $\nu=0.2$. 
We see that the bound can be severely violated under certain (albeit rare) circumstances.}
\label{fig:asymptotic}
\end{center}
\end{figure}

\subsection{Thermodynamic inequality}

Let us now come to the finite-time bound. 
The latter is proportional to the mutual information which, using our new insight on the distribution of price trajectories, can be formulated as
\be
\begin{aligned}
I(Y; X_{1:n}) & = \sum_{y\in\{0, 1\}} p(y)\int d\beta\; p(\beta|Y=y)\log\frac{p(\beta|Y=y)}{p(\beta)}
\\
& = H(Y) - H(Y | X_{1:n}) \; ,
\end{aligned}
\ee
in terms of the conditional entropy
\be
H(Y | X_{1:n}) = -\theta\sum_{k=0}^n C^n_k q^k(1-q)^{n-k}\log\left({\cal L}_q^{2k-n}[\theta]\right) - (1-\theta) \sum_{k=0}^n C^n_k(1-q)^k q^{n-k}\log\left(1- {\cal L}_q^{2k-n}[\theta]\right) \; .
\ee
Alternatively, we have seen that mutual information can also be computed using the chain rule over conditional components
\be
I(Y; X_{1:n}) = \sum_{i=1}^nI(Y;X_i|X_{1:i-1}) = \sum_{i=1}^n\mathbb{E}[I(Y;X_i | \beta_i)]
\ee
in terms of the stochastic conditional mutual information 
\be
I(Y; X_{i+1} | \beta) = h\left(q\frac{{\cal L}_q^{2\beta-i}[\theta]}{{\cal L}_q^{2\beta-i+1}[\theta]}\right) - h(q) \; .
\ee
New information acquired at step $i+1$ follows by averaging:
\be
I(Y;X_{i+1}|X_{1:i}) + h(q) = \theta\sum_{k=0}^iC^i_kq^k(1-q)^{i-k}\frac{1}{{\cal L}_q^{2k-i}[\theta]}h\left(q\frac{{\cal L}_q^{2k-i}[\theta]}{{\cal L}_q^{2k-i+1}[\theta]}\right) \; .
\ee
An illustration of the finite-time bound holding is provided in Figure \ref{fig:gm_gain}.
We also find that gain fluctuations are quite significant.
In particular, it appears that $\mathbb{V}[G_n] > \sum_{i=1}^n \mathbb{V}[\mu_i]$ which means that payoffs are positively autocorrelated.
This can be checked numerically by computing the joint expectation
\be
\begin{aligned}
\mathbb{E}[\mu_n\mu_{n-i}] & = (1-q)^2\sum_{k=0}^n\sum_{j=0}^{n-i}p(\beta_n=k | \beta_{n-i}=j)p(\beta_{n-i}=j)s_n(k)s_{n-i}(j)
\\
& = (1-q)^2\sum_{k=0}^n\sum_{j=\max(k-i, 0)}^{\min(k, n-i)} C^i_{k-j}q^{k-j}(1-q)^{i-(k-j)}C^{n-i}_jq^j(1-q)^{n-i-j} s_n(k)s_{n-i}(j)
\\
& = (1-q)^2\sum_{k=0}^n q^k(1-q)^{n-k}s_n(k)\sum_{j=\max(k-i, 0)}^{\min(k, n-i)}C^i_{k-j}C^{n-i}_js_{n-i}(j) \; .
\end{aligned}
\ee
Finally, the asymptotic bound of Touzo et al. (recalled in the Corollary) is displayed in Figure \ref{fig:asymptotic}. 
The bound becomes tight as $\nu\to0$, though this is blurred by finite-size numerical effects for the lowest value of $\nu$ we considered.
That the (asymptotic) bound holds only in expectation is well seen by plotting the histogram of cumulative payoffs, as in the inset of Figure \ref{fig:asymptotic}. 
Better quantifying the magnitude of gain fluctuations could perhaps allow turning Theorem \ref{th} into a more general fluctuation theorem \cite{Seifert}.

\section{Discussion}
\label{sec:Concl}

To summarize, we have proven that the expected achievable gain of the informed trader at any point in time can be upper bounded by the amount of private information remaining at his disposal 
(or equivalently the remaining level of ignorance of the market maker).
The reasoning behind the proof can be reconstructed with the following chain of results:
\be
\left(P1 \cup P2 \cup L2\right) \Rightarrow L1 \Rightarrow Theorem \Rightarrow Corollary
\ee
\be
P3 \cup P4 \Rightarrow L2
\ee
with propositions P1-4 proven in stand-alone.
The proof exploits the structural invariance of the game at each step and relies on an inequality solely involving the binomial entropy function and its derivatives. 
This result generalizes that obtained in \cite{Touzo} and also allows bypassing the more technical considerations based on infinite sums used in section 3 of their paper to prove their result.

In this work, we made no distinction between a game composed of a population of informed traders (representing a fraction $\nu$ of the total population) and a single trader
acting with frequency $\nu$. While one could argue that $\nu$ can be thought of as a free parameter in the latter case but not in the former, the results proven in this paper 
hold regardless. A related, yet somewhat different, situation is to consider a game without any noise traders, such that the informed trader is forced to dilute
his informational advantage by only using it with probability $\nu$. Assuming the market maker still manages to break even in this case, the expected gain of the informed trader
should be zero (since the game is zero-sum). But this may no longer be true if the market maker has to learn the value of $\nu$, which is another sensible extension of the model.

In the setting considered in this paper, the maximum achievable gain by the informed trader diverges as $\nu\to0$ (since $T\sim \nu^{-1}$). 
In that limit, the measurement process becomes reversible, preventing the market maker from inferring any information from the order sequence. 
The caveat is that it takes an infinite amount of time $\tau\sim\nu^{-2}$ to cash in. 
Maximizing the expected gain in a finite time $\tau$, e.g. a terminal time at which the value of the asset shall be made public, introduces a tradeoff between 
minimizing information disclosure (small $\nu$) and posting enough trades (large $\nu$). 
As suggested by the above scaling (and supported by numerical evidence), $\nu\sim\tau^{-1/2}$ seems like the optimal strategy.

There are many other interesting research directions.
One would be to understand how the above picture changes if the market maker becomes risk-averse and thus no longer accepts bearing risk without any expected payoff in return. 
Another obvious one would be to investigate whether our result can be extended to a setting where the asset value is no longer static (as in \cite{Benz}).
Including a second asset correlated to the first could also allow studying how the incorporation of information in prices is blurred when informed traders pursue multiple
objectives (a canonical example being risk control) \cite{Lost}.
More generally, understanding how information diffuses among asymmetrically informed market participants from the perspective of information thermodynamics certainly seems like a direction worth exploring further.

\end{document}